\documentclass[12pt, draftclsnofoot, onecolumn]{IEEEtran}

\normalsize 
 
\ifCLASSINFOpdf
\else
\fi

\usepackage{algorithmic}
\usepackage[linesnumbered,lined,ruled]{algorithm2e}
\usepackage[latin5]{inputenc}
\usepackage{graphicx}
\usepackage{epsfig,setspace,times}
\usepackage{epstopdf} 
\usepackage{color}
\usepackage{multirow}
\usepackage{cite}
\usepackage[cmex10]{amsmath}
\usepackage{amssymb}
\usepackage{algorithmic}
\usepackage[linesnumbered,lined,ruled]{algorithm2e}
\usepackage{array,booktabs}
\usepackage{blkarray}
\usepackage[tight,footnotesize]{subfigure}
\usepackage{bm}
\usepackage{amsthm}
 \usepackage{makecell}

\theoremstyle{definition}
\newtheorem{thm}{Theorem}
\newtheorem{lem}{Lemma}

\begin{document}

\title{Optimal Offline Packet Scheduling in Energy Harvesting 2-user Multiple Access Channel with Common Data}

\author{Burhan~Gulbahar,~\IEEEmembership{Senior Member,~IEEE}
\thanks{B. Gulbahar is a faculty member at the Department of Electrical and Electronics Engineering,  \"{O}zye\u{g}in University, \.{I}stanbul, Turkey 34794 and the coordinator of Vestel Electronics Academy Graduate Education Program,  Manisa, Turkey 45030. e-mail: burhan.gulbahar@ozyegin.edu.tr.}
}
 
\markboth{Submitted for publication}%
{Submitted paper}

 
 
\maketitle

\vspace*{-0.5cm}

\begin{abstract}
The lifetime and the sustainability of the  wireless sensor networks (WSNs) can be increased with energy harvesting transmitters utilizing optimum packet scheduling. On the other hand, WSNs are observed to collect spatially or temporally correlated data which should be taken into account for the optimum packet scheduling in an energy harvesting system. However, the solutions available for 2-user multiple-access channel  (MAC) systems with energy harvesting transmitters do not consider the common data or the correlation among the data. In this paper, optimal packet scheduling for energy harvesting 2-user Gaussian MAC with common data is achieved by assuming deterministic knowledge of the data and energy packets, i.e., offline solution. The optimum departure region is found by using Karush-Kuhn-Tucker (KKT) conditions generalizing the solutions obtained for the MAC without common data. An efficient iterative backward water-filling algorithm is defined. The optimum solution is numerically compared  with the case of no scheduling,  uniform power scheduling and the previous solutions defined for the MAC without common data by showing the improvement obtained with the optimization. 

\end{abstract}
 

%
\IEEEpeerreviewmaketitle

\section{Introduction}
 
\IEEEPARstart{E}{nergy} harvesting with optimum  packet scheduling policy is significantly important for increasing the lifetime and sustainability of wireless sensor networks (WSNs) and to achieve the demands of green communications  \cite{b4, b11, b12, b18}.  The scarcity and sporadic availability of the energy make it necessary to store it and utilize optimally. Therefore, optimum power management and data transfer schemes are significantly important for WSNs. The correlation of data observed in WSNs is one of the most important factors to be taken into account for designing optimum power scheduling algorithms for the energy harvesting transmitters in order to consume the available resources more efficiently \cite{b24}. In this article, previous optimal solutions defined for Gaussian MAC are extended to include common data observed at the energy harvesting transmitters.

Optimum online and offline packet scheduling in energy harvesting communication systems are recently investigated for single hop, multiple-access channel (MAC) and broadcast systems. Optimal packet scheduling for single-user energy harvesting communication systems are presented in \cite{b7, b8}. In \cite{b1}, a directional water-filling algorithm optimizes the throughput for a single-user fading channel with additive Gaussian noise with finite capacity rechargeable batteries under offline and online knowledge. In \cite{b14}, a two-hop relaying communication network with energy harvesting rechargeable nodes  is  formulated for the offline end-to-end throughput maximization as a convex optimization problem.

Besides that, similar analyses are achieved for the MAC schemes. In \cite{b4, b11, b10}, optimal packet scheduling problem is solved in a 2-user MAC system with energy harvesting transmitters where the energy harvesting times and harvested energy amounts are known before the transmission. The Karush-Kuhn-Tucker (KKT) solution and the generalized iterative backward water-filling  algorithm are presented. However, these studies do not consider common data and correlation among the collected sensor data which is the main contribution achieved by the article. In addition, the optimum policy is not compared with no power scheduling case and uniform power scheduling polices. 

Furthermore, in \cite{b3}, the proposed optimum scheduling policies are extended by including one-way energy transfer capability between two transmitter nodes. Moreover, in \cite{b13}, the capacity region of Gaussian MAC with  amplitude constraints and batteryless energy harvesting transmitters are analysed. On the other hand, in \cite{b12}, optimal continuous-time online power policies for energy harvesting MACs are presented. In \cite{b16}, energy harvesting transmitter  and receiver pair is considered in a utility maximization framework achieving power policy using a water-filling approach. In \cite{b18}, optimal transmit power policy for energy harvesting transmitters  in a Gaussian MAC is presented by also considering storage losses. However, these studies do not consider the common data in an optimum packet scheduling framework. 

On the other hand, data correlation in WSNs are significantly important  to save the power-bandwidth resources \cite{b24}. In \cite{b2}, the explicit characterization of the capacity region in a Gaussian MAC channel with common data and fading is considered.  The optimum power allocation achieving the arbitrary rate tuples on the boundary of the capacity region are presented and  numerically computed. However, the study does not consider optimum power scheduling and operates only in a single time interval. The current study extends the optimum solution to include multiple time intervals in an optimum power scheduling policy as the main contribution. 

Moreover, in \cite{b20}, the capacity region of the discrete p-transmitter/q-receiver MAC defined as General MAC (GMAC) with a common message is derived as a generalization of \cite{b2}. In \cite{b21}, information-theoretic results and power allocation policies  in combination with joint source-channel codes on the transmission of memoryless dependent sources through a memoryless fading MAC are analysed. In \cite{b22},  2-user  MAC with common message (MACCM)   and MAC with conferencing encoders (MACCE)  with channel state information (CSI) are analyzed. The capacity results for the Gaussian MAC with cooperative encoders and with additive interference known  non-causally to both encoders are presented. However, none of these studies combine energy harvesting and optimum packet scheduling in a Gaussian MAC with common data.
 
To the best of our knowledge, in this work, optimum offline packet scheduling solution, for the first time, is given for Gaussian MAC with common data and energy harvesting transmitters. The KKT solution is given, for the first time,  for optimum packet scheduling problem for Gaussian MAC with common data and energy harvesting transmitters. An efficient iterative water-filling algorithm is introduced for the optimum solution although the optimum solution and the determination of the water levels are more complicated compared with the MAC without common data. The departure region or the capacity boundary surface for the data rates of the individual and the common data messages is numerically simulated by using the defined optimum water-filling algorithm. The optimum packet scheduling solution is compared with the solutions defined for Gaussian MAC without common data, the case for no power scheduling and uniform power scheduling framework by showing the advantages of the proposed solution.  

The remainder of the paper is organized as the following. In Section \ref{sysmodel}, power scheduling policy and the system model for the Gaussian MAC with common data are defined. Then, in Section \ref{ops}, data throughput maximization problem is defined. In Section \ref{optw}, KKT solution for the defined problem is proposed and efficient iterative water-filling algorithm is presented. In Section \ref{simulr1}, a simulation study is performed illustrating optimum scheduling policy, the departure region boundary surfaces and the comparison of the proposed solution with no scheduling and uniform power scheduling cases, and the comparison  with the previous solutions defined for MAC without common data. Then, in Section \ref{fw}, future work and open issues are discussed. Finally, in Section \ref{conc}, the conclusions are given.

\section{System Model}
\label{sysmodel}

\begin{figure}[t!]
\begin{center}
\includegraphics[width=3in]{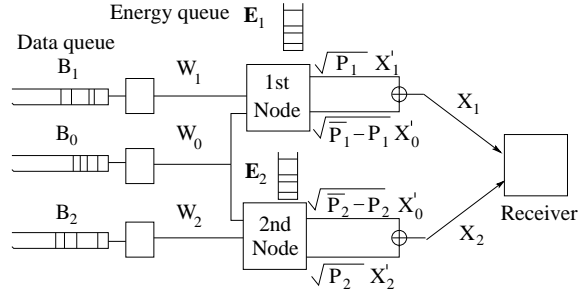} \\ 
\caption{System model for Gaussian MAC with common data and energy harvesting transmitters.}
\label{fig1}
\end{center}
\end{figure}

In this paper, energy harvesting Gaussian MAC with two transmitters and one receiver is considered as shown in Fig. \ref{fig1} \cite{b4, b11} while incorporating the common data model given in \cite{b2}. Each user has their individual data packets and also a common message known by both transmitters. It is assumed that the amounts of harvested energy and the harvesting times are known before the data transmission. Similar to  \cite{b4, b11}, energy harvesting times are put in ascending order and the length of the time interval between two energy harvesting instants $t_n $ and $t_{n+1}$ is denoted by $L(n) $  while energy harvesting starts at $t_1$ and ends at $t_N$. It is assumed that the final deadline time instant to transmit data bits is $T_f$ with the $N$th time interval length being equal to $L(N)= T_f - t_{N}$. For example, 1st user harvests $E_1(n) $ at the time instant $t_n$ and 2nd user harvests $E_2(n+1)$ at the time instant $t_{n+1}$ possibly both users harvesting energy at any single time instant. An illustrative energy harvesting scenario is shown  in Fig. \ref{fig0}.

\begin{figure}[t!]
\begin{center}
\includegraphics[width=3in]{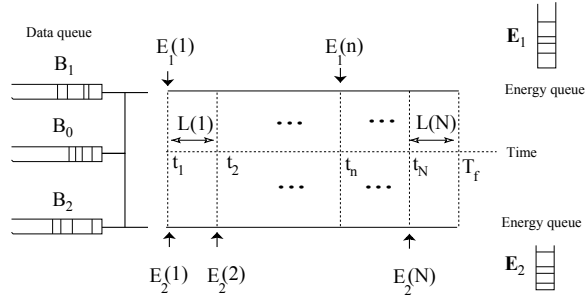} \\ 
\caption{Energy harvesting system model with common data where the data packets are available at the beginning.}
\label{fig0}
\end{center}
\end{figure}

In this article, stochastic nature of energy harvesting and the time-varying channel states are not formulated due to complexity of the issue. The harvested energy is stored in sensor nodes and the problem is simplified by assuming that the data packets are available before the transmission and information about the energy harvesting times is available. Therefore, in this article we consider an offline solution with a deterministic system setting instead of an online solution with stochastic energy, data and channel states leaving the consideration as a future work. The target is to find the maximum data throughput regions, i.e., maximum departure region \cite{b11}, for three independent messages, i.e., $W_0$, $W_1$ and $W_2$, where $W_0$ is known by both the users, for any given deadline time $T_f$ and propose a water-filling algorithm finding the optimal solution.   

The system model for the Gaussian MAC is shown in Fig. \ref{fig1}. The data packets $(B_0, B_1, B_2)$ are available at the beginning before the transmission. The packets are modulated into message sequences $W_{0,1,2}$ and 1st node  knows both $W_0$ and $W_1$ while 2nd node  knows both $W_0$ and $W_2$. The transmitted symbols, i.e., $X_{k}$, is a function of $W_{0,k}$ for $k \in [1,2]$. Each node $k$ for $k \in [1,2]$ has the available power $\overline{P_k} $ in units of (W) for transmission. Each symbol consists of the addition of the symbols for the independent message, i.e., $X_{1,2}^{'}$, and the common message, i.e., $X_0^{'}$. Each node $k$ assigns the power levels $P_k$ to its own symbol for the independent message and $\overline{P_k} - P_i$ to the symbol for the common message. In the receiver, a beam-forming gain occurs for the common message such that $P_0 = \left( \sqrt{\overline{P_1} - P_1} + \sqrt{\overline{P_2} - P_1} \right)^2$ is larger than $\overline{P_1} - P_1 + \overline{P_2} - P_2$. The inputs and the output at some specific time are related as follows,
\begin{eqnarray}
        \label{eq1} 
X_k  &=& \sqrt{ P_k} X_i^{'}  + \sqrt{\overline{P_k} - P_k} X_0^{'}, \, k \in \left[ 1, 2\right] \\
Y  &=& H_1 \,X_1 +  H_2 X_2 + Z  \nonumber  \\     
  \label{eq3} 
  & =&  h \, \left(  \sqrt{ P_1} X_1^{'} +   \sqrt{ P_2} X_2^{'} +  \sqrt{ P_0} X_0^{'} \right) + Z  
\end{eqnarray}
where $H_1 = H_2 = h$ and $Z$ is a zero mean Gaussian noise sample.   The  aim is to formulate the effect of common data on optimum power scheduling in energy harvesting MAC in a deterministic setting and the fading channel coefficient is assumed constant during transmission. 

The capacity of the channel in a time interval $L$ with the total power $P$ is denoted by $C(P)$ and given by the following,
\begin{equation}
C(P) = W_{Tot} \, \mbox{log} \left( 1 + \frac{P \, h}{  W_{Tot} N_0 } \right)  
\end{equation}
where \textit{log} denotes base $2$ logarithm, $W_{Tot}$ (Hz) refers the total bandwidth, $N_0$ (W/Hz) is the noise spectral density and $h$ is the fixed path loss. Throughout the article, $W_{Tot} \, N_0 \, / \, h$ is denoted by the power constant $A $ (W) such that $C(P) = W_{Tot}  \, \mbox{log} ( 1 + P \, / \, A ) $. The total transmitted bits in the time interval $L$ can be represented by $B(P, L) = C(P) \times L$. The capacity region of the Gaussian MAC with the common data is given by the following \cite{b2},
 \begin{eqnarray}         
R_1   &\leq&    C(P_1)   \nonumber \\
 R_2   &\leq& C(P_2)  \nonumber \\
  R_1 + R_2  &\leq&  C(P_1 + P_2) \nonumber \\
    \label{eq6} 
 R_1 + R_2 + R_0   &\leq&   C(P_1 + P_2 + P_0) 
 \end{eqnarray}
where $0 \leq P_1 \leq \overline{P_1}$, $ 0 \leq P_2 \leq \overline{P_2}$ and $P_0   =   \left( \sqrt{\overline{P_1} - P_1} +  \sqrt{\overline{P_2} - P_2}  \right)^2$,   $\overline{P_1} = P_1 + \rho^2 P_0$, $\overline{P_2} = P_2 + (1 - \rho)^2 P_0$, $\rho$ is the variable adjusting the contribution to $P_0$ by each node and $0 \leq \rho \leq 1$. A rate triplet $(R_0, R_1, R_2)$ is achievable if a sequence of $((2^{n R_0}, 2^{n R_1}, 2^{n R_2}), n)$ codes exist where the average probability of error for decoding messages correctly approaches zero as $n$ goes to infinity \cite{b2}. The capacity region $R(P_1, P_2)$ is defined by the closure of the set of achievable $(R_0, R_1, R_2)$ rate triplets. The form of $R(P_1, P_2)$ is illustrated in Fig. \ref{fig2} for some specific $P_1  \leq \overline{P_1}$ and $P_2 \leq \overline{P_2}$. 
\begin{figure}[t!]
\begin{center}
\includegraphics[width=3in]{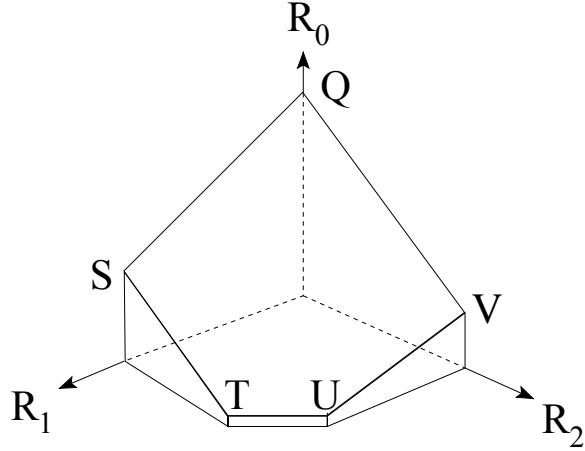} \\ 
\caption{Illustrative capacity region of Gaussian MAC with common data.}
\label{fig2}
\end{center}
\end{figure}
The data rates of main boundary points $(Q, S, T, V, U)$ on the three dimensional (3D) capacity curve are listed as follows \cite{b2}, 
\begin{eqnarray}    
S &:& (C(P_1), 0, C(P^s_{0,2}) - C(P_1)) \nonumber \\
T &:& (C(P_1), C(P^s_{1,2}) - C(P_1), C(P^s_{0,2}) - C(P^s_{1,2}))  \nonumber \\
U &:& (C(P^s_{1,2}) - C(P_2), C(P_2), C(P^s_{0,2}) - C(P^s_{1,2}))\nonumber \\
V &:& (0,C(P_2), C(P^s_{0,2}) - C(P_2)) \nonumber \\  
\label{eq12}
Q &:& (0,0, C(P^s_{0,2})  
\end{eqnarray}
where $P^s_{0,2}$ = $P_0 + P_1 + P_2$, $P^s_{1,2}$ = $P_1 + P_2$.
It is proved in \cite{b2} that all the points on the capacity region of Gaussian MAC are achieved by some point on the line segment $T-U$ of $R(P_1, P_2)$ for some $0 \leq P_1 \leq \overline{P_1}$, $ 0 \leq P_2 \leq \overline{P_2}$. Furthermore, the union $\bigcup_{(P_1, P_2, P_0, \rho) \in F} R_f(P_1, P_2, P_0, \rho )$  is denoted by the capacity region $R(\overline{P_1}, \overline{P_2})$ where the set of achievable $(R_0, R_1, R_2)$ in (\ref{eq6}) is given by $R_f(P_1, P_2, P_0, \rho ) \equiv R(P_1, P_2)$ and $F$ is the following where $k \equiv \{0,\, 1, \,2 \}$,
\begin{equation}
\label{eq14a} 
F  =  \Big\lbrace (P_{k}, \rho ):  P_{k} \geq 0, 0 \leq \rho \leq 1, P_1 + \rho^2 P_0 \leq \overline{P_1},  P_2 + \big(1 - \rho\big)^2 P_0 \leq \overline{P_2} \Big\rbrace 
\end{equation}
It is stressed out that the boundary surface of $R(\overline{P_1}, \overline{P_2})$ can be found with the following optimization problem by varying rewarding values $\mu$,
\begin{eqnarray}
   \label{eq14} 
 \max_{\boldsymbol{R}, P_1, P_2, P_0, \rho} \mu_1 R_1 + \mu_2 R_2 + \mu_0 R_0 \mbox{  s.t.  } \boldsymbol{R} \in R_f(P_1, P_2, P_0, \rho )
\end{eqnarray}
It is observed in \cite{b2} that various regions of $\mu$ values achieve  $R(\overline{P_1}, \overline{P_2})$ at the defined boundary points as in the following,
\begin{eqnarray}
\mu_1 \geq \mu_0  \geq \mu_2 &\rightarrow &  \, \, S \nonumber \\
\mu_1 \geq \mu_2  \geq \mu_0  &\rightarrow&    \, \, T   \nonumber \\
\mu_2 \geq \mu_1  \geq \mu_0  &\rightarrow&   \, \, U  \nonumber \\
\mu_2 \geq \mu_0  \geq \mu_1 &\rightarrow&   \, \, V \nonumber \\
\label{eq17} 
\mu_0 \geq \mbox{max} (\mu_1, \mu_2) &\rightarrow &  \, \, Q  
\end{eqnarray}
Therefore, an optimization framework can cover the whole capacity region by optimizing the solution at the defined boundary points. These points are extended to include multiple time intervals in the following sections.

Now, after defining the system model and capacity region for Gaussian MAC with common data, the capacity maximization  convex problem and its  solution are analysed.
\section{Data Throughput Maximization Problem}
\label{ops}

In this section, firstly, the data throughput or the maximum departure region under a deadline time constraint $T_f$ is defined. Then, the convexity of the region is proven and the capacity maximization problem is defined by using Lagrange multipliers for the defined convex optimization formulation. The problem formulation is achieved by using KKT conditions \cite{b26}. 

The departure region for the overall harvesting duration is characterized with  \textit{Lemma \ref{lemma1}} which can be proved by using the similar approaches in \cite{b2}, \cite{b11} and\cite{b25} regarding the ergodic capacity region for the overall transmitted bits within a finite amount of time.  
\begin{lem} 
\label{lemma1}
The departure region in $N$ time intervals for a time-nonvarying fading Gaussian MAC with common data and harvested energies denoted by the vectors $(\boldsymbol{\underline{E_1}}, \boldsymbol{\underline{E_2}})$ is given by
\begin{eqnarray}
   \label{lemma1_1} 
 B_d(\boldsymbol{\underline{E_1}}, \boldsymbol{\underline{E_2}}, N) = \bigcup_{\boldsymbol{P_0}, \boldsymbol{P_1}, \boldsymbol{P_2}, \boldsymbol{\rho} \, \in F_N} B(\boldsymbol{P_0}, \boldsymbol{P_1}, \boldsymbol{P_2}, \boldsymbol{\rho})  
\end{eqnarray}   
where $B(\boldsymbol{P_0}, \boldsymbol{P_1}, \boldsymbol{P_2}, \boldsymbol{\rho})$ is the set of departure triplets ($B_0, B_1, B_2$) s.t.
\begin{eqnarray}
B_1   & \leq &  \sum_{n = 1}^{N}  C\big(P_1(n)\big)\, L(n) \nonumber \\  
B_2    & \leq &  \sum_{n = 1}^{N} C\big(P_2(n)\big)\, L(n) \nonumber \\ 
B_1 \, + \, B_2   & \leq &  \sum_{n = 1}^{N}  C\big(P_1(n) + P_2(n)\big) \, L(n)  \nonumber \\
\label{lemma1_2} 
B_0 \, + \, B_1 \, + \,B_2     &\leq &  \sum_{n = 1}^{N}    C\big(P_0(n) + P_1(n) + P_2(n)\big)  \, L(n)   
\end{eqnarray}
where  $ F_N = \big \lbrace  (\boldsymbol{P_0}, \boldsymbol{P_1}, \boldsymbol{P_2}, \boldsymbol{\rho}): \, P_0(n), P_1(n), P_2(n) \, \geq \, 0, \, 0 \, \leq  \, \rho(n) \, \leq 1, \, n  \in [1, N] \big \rbrace $, $\boldsymbol{\underline{E_{1,2}}} = \left[ \underline{E_{1,2}}(1) \, \underline{E_{1,2}}(2) \hdots \underline{E_{1,2}}(N)\right]$   with $\underline{E_{1,2}}(n)$ denoting the total harvested energy until the $n$th time interval and $E_{1,2}(1)$ is the energy available at the beginning, $\underline{E_{1,2}}(N)$ is the total harvested energy to be consumed until the final deadline time $T_f$ , $\overline{P_1}(n)$ and $\overline{P_2}(n)$ are the assigned power levels for the 1st and 2nd nodes, respectively, satisfying $\, P_1(n) + \rho^2(n) P_0(n) \leq  \overline{P_1}(n)$,  $P_2(n) + \big(1 -\rho(n)\big)^2 P_0(n) \leq \overline{P_2}(n)$  while the following causality conditions hold for the total consumed energy regarding its relation to the harvested energy for  $n = \left[1 , \, N-1\right]$,
\begin{eqnarray}       
\sum_{n = 1}^{N} \overline{P_1}(n) L(n) \, - \, \underline{E_1}(N)  &  =  &   0   \nonumber \\
\sum_{n = 1}^{N} \overline{P_2}(n) L(n) \, - \, \underline{E_2}(N)    &  =  &  0     \nonumber \\
\sum_{j = 1}^{n} \overline{P_1}(j) L(j) \, - \, \underline{E_1}(n)  &  \leq  &  0 \nonumber \\
\label{eq22}
\sum_{j = 1}^{n} \overline{P_2}(j) L(j) \, - \, \underline{E_2}(n)   &  \leq  &  0    
\end{eqnarray}
\end{lem} 

\begin{proof}
The approaches similar to the study in \cite{b11} developed for the MAC without common data can be used, especially \textit{Lemma 1} and \textit{2}. It can be easily proved that the optimal power allocation policy does not change the transmission rate or power during energy harvesting times. Furthermore, the rate triplets $(B_0, B_1, B_2)$ given in (\ref{lemma1_2}) define the maximum departure region for any feasible power policy $(\boldsymbol{P_0}, \boldsymbol{P_1}, \boldsymbol{P_2}, \boldsymbol{\rho}) \, \in F_N$ by recursively combining the individual three dimensional (3D) capacity regions for each time interval instead of 2D capacity region combinations in \cite{b11}. In other words, for the first time interval, the 3D capacity region $B(\overline{P_1}(1), \overline{P_2}(1))$ gives the maximum departure region. Then, for the second time interval, any point on the capacity region of the first time interval can be taken as the origin and the overall capacity region for the total of first two time intervals are found by combining two of the regions as shown in (\ref{lemma1_2}). It can be applied recursively to the next time intervals.
\end{proof}
 
By using the similar approach for the capacity region of Gaussian MAC with common data \cite{b2}, the capacity region is unchanged if  the inequalities in $\, P_1(n) + \rho^2(n) P_0(n) \leq  \overline{P_1}(n)$,  $P_2(n) + \big(1 -\rho(n)\big)^2 P_0(n) \leq \overline{P_2}(n)$  are changed with equality. Furthermore, the departure region points on the boundary surface of $B_d(\boldsymbol{\underline{E_1}}, \boldsymbol{\underline{E_2}}, N)$ is the set of all departure triplets $ \boldsymbol{B} =  (B_0, B_1, B_2)$ such that $ \boldsymbol{B}$ is a solution to the following problem for some value  of $(\mu_0, \mu_1, \mu_2) \in \mathbb{R}_{+}^{3}$ \cite{b2, b11},
\begin{equation}
   \label{eqoptprob} 
 \max_{\boldsymbol{B}, \boldsymbol{\widetilde{E_1}}, \boldsymbol{\widetilde{E_2}}} \sum_{k = 0}^{2}\mu_k B_k 
  \mbox{  s.t.  } ( \boldsymbol{B}, \boldsymbol{\widetilde{E_1}}, \boldsymbol{\widetilde{E_2}} ) \in \mathbb{P} 
\end{equation}
where $ \mathbb{P} =  \big \lbrace (\boldsymbol{B}, \boldsymbol{\widetilde{E_1}}, \boldsymbol{\widetilde{E_2}} ) : \boldsymbol{\widetilde{E_1}}, \boldsymbol{\widetilde{E_2}} \in \mathbb{R}_{+}^{N}, \boldsymbol{B} \in B_d(\boldsymbol{\widetilde{E_1}}, \boldsymbol{\widetilde{E_2}}, N)  \big \rbrace  $,  $ P_{k}(n)  \geq \, 0, \, 0 \, \leq  \, \rho(n) \, \leq 1$, the conditions in (\ref{eq22}) are satisfied by replacing $\overline{P_{1,2}}(n)$ with $\widetilde{P_{1,2}}(n)$ and $\underline{E_{1,2}}(m)$  with $\widetilde{E_{1,2}}(m)$, $\widetilde{E_{1,2}}(m) \leq \underline{E_{1,2}}(m)$ and $\widetilde{E_{1,2}}(N) = \underline{E_{1,2}}(N)$ for $n  \in [1, N]$, $m \in [2, N]$ and $k \in [0, 2]$. Therefore, by maximizing  $\sum_{k = 0}^{2}\mu_k B_k  $ for specific $\widetilde{E_{1,2}}(n)$ the points on the boundary region can be achieved. Then, we can prove that the departure region is convex.

\begin{lem} 
\label{lemma2} 
$B_d(\boldsymbol{\widetilde{E_1}}, \boldsymbol{\widetilde{E_2}}, N)$ is a convex region and  $\mathbb{P}$ is a convex set.
\end{lem} 
The proof of Lemma \ref{lemma2} is given in  Appendix \ref{app0}.
 
By using the the convexity, KKT conditions can be utilized to transform the problem into more convenient forms to be solved as the following \cite{b26},
\begin{eqnarray}
   \label{eqoptprob2} 
  \min_{(\boldsymbol{B},  \boldsymbol{P}, \boldsymbol{\rho})}  \, \, \,    - \sum_{k = 0}^{2}  \mu_k B_k + \sum_{i = 1}^{N} \lambda_{1,i} \left(   \sum_{n = 1}^{i} \bigg( P_1(n) + \rho^2(n) P_0(n) \bigg) L(n)   -   \underline{E_1}(i)  \right)    && \nonumber \\  
 +  \sum_{i = 1}^{N} \lambda_{2,i} \left(  \sum_{n = 1}^{i} \bigg(  P_2(n) + \big( 1 -  \rho(n)\big)^2 P_0(n) \bigg) L(n) \, - \,   \underline{E_2}(i)  \right)  \, \, \, \, \, \, \, \, \, \, \, \, \, \, \, \,   \, \, \, \, \, \, \, \, \, \, \, \, \, \, \, \, \, \, \, \, \, \, \, && \nonumber \\ 
 + \sum_{i = 1}^{N}    \bigg(  - \lambda_{3,i}   \rho(i)   + \lambda_{4,i} \big(   \rho(i)  -1 \big) - \lambda_{5,i}   P_1(i)     - \lambda_{6,i}  P_2(i)  -  \lambda_{7,i}   P_0(i)  \bigg)    L(i)   \, \, \, &&
\end{eqnarray}
where $B_k$ corresponds to the extension of the boundary points defined in (\ref{eq12}) to the $N$ time intervals, e.g., point $T$ for $\mu_1 \geq \mu_2 \geq \mu_0$ s.t. $B_0 = \sum_{n = 1}^{N} C(P^s_{0,2}(n)) - C\big(P^s_{1,2}(n)\big)$, $B_1 = \sum_{n = 1}^{N} C\big(P_1(n)\big)$, $B_2 = \sum_{n = 1}^{N} C\big(P^s_{1,2}(n)\big) - C\big(P_1(n)\big)$, and the following constraints are satisfied for  $n \in  [1,N]$, $ m \in [1,N-1]$,  $p \in [3, 7]$ and $k \in [0, 2]$,
\begin{eqnarray}
   \label{c1_eq}
    P_{k}(n) \geq 0, \, 0 \leq \rho(n) \leq 1,  \, \lambda_{p,n} \geq 0, \,   \lambda_{1,m} \geq 0,   \, \lambda_{2,m} \geq 0  && \\
   \label{c2_eq}   
  \sum_{n = 1}^{N} \bigg( P_1(n) + \rho^2(n) P_0(n) \bigg) L(n) \, - \, \underline{E_1}(N)   =  0   && \\
   \label{c3_eq}    
   \sum_{n = 1}^{N} \bigg( P_2(n) + \big( 1 -  \rho(n)\big)^2 P_0(n) \bigg) L(n) \, - \, \underline{E_2}(N)  = 0 && \\
   \label{c4_eq}    
     \sum_{n = 1}^{m} \bigg( P_1(n) + \rho^2(n) P_0(n) \bigg) L(n) \, - \, \underline{E_1}(m) \leq 0  &&  \\
   \label{c4_eq1}    
     \sum_{n = 1}^{m} \bigg( P_2(n) + \big( 1 -  \rho(n) \big)^2 P_0(n) \bigg) L(n) \, - \, \underline{E_2}(m) \leq 0  &&  \\
      \label{c5_eq}    
    \lambda_{1,m} \Bigg( \sum_{n = 1}^{m} \bigg( P_1(n) + \rho^2(n) P_0(n) \bigg) L(n) \, - \, \underline{E_1}(m)  \Bigg) = 0 &&  \\
        \label{c5_eq1}  
      \lambda_{2,m} \Bigg( \sum_{n = 1}^{m} \bigg( P_2(n) + \big( 1 -  \rho(n) \big)^2 P_0(n) \bigg) L(n) \, - \, \underline{E_2}(m) \Bigg) = 0 &&  \\
      \label{c5_eq2a}  
      \lambda_{3,n}  \, \rho(n)  =0; \, \, \lambda_{4,n} \big(   \rho(n)  -1 \big) = 0; &&  \\  
      \label{c5_eq2} 
      \, \,\lambda_{5,n} P_1(n)=0;  \, \, \lambda_{6,n} P_2(n)=0;  \, \, \lambda_{7,n} P_0(n)=0    &&  
\end{eqnarray}
Taking the derivative with respect to  $\rho(n)$ for $n \in [1, n]$ and equalizing to zero gives the following,  
\begin{eqnarray}
       \label{eq53} 
  \rho(n) &\mbox{:}&   P_0(n)\left[ \rho(n) \left(2  \lambda_{1,n}^{p} + 2  \lambda_{2,n}^{p} \right) \right]  
  =  P_0(n) \left[ 2  \lambda_{2,n}^{p}   \right] +   \lambda_{3,n} - \lambda_{4,n}    
 \end{eqnarray}
where $\lambda_{s,n}^{p} =  \sum_{i = n}^{N} \lambda_{s,i} $ for $s \in [1, 2]$. If $0 < \rho(n) < 1$ and $P_0(n) > 0$, the quadratic expression can be removed from the Lagrangian function by calculating $\rho(n)$ as the following,
\begin{equation}
       \label{eq54} 
\rho(n) = \frac{ \lambda_{2,n}^{p}   }{\lambda_{2,n}^{p} +   \lambda_{1,n}^{p}  }
\end{equation}
For the other cases where $\rho(n) = 0$ or $\rho(n) = 1$ and $P_0(n) > 0$, the quadratic terms similarly disappear leading to a convex set of equations. 

Now, the problem in (\ref{eqoptprob2}) can be solved as a convex optimization problem leading to the unique global optimum solution. The solution is found by solving KKT optimality conditions with unique KKT multipliers for the global solution.

\section{Optimal Scheduling Solution and Water-Filling Algorithm }
\label{optw}

$(B_0, B_1, B_2 )$ departure triplets are found by  varying the rewarding values of $(\mu_1, \mu_2, \mu_0)$ such that $\mu_1 \geq  \mu_2 \geq  \mu_0 $, $\mu_1 \geq  \mu_0 \geq  \mu_2 $ and $\mu_0 \geq  \mbox{max}(\mu_1, \mu_2)$ are utilized for the half part of the capacity region. The other half part can be found by changing the roles of the nodes and $\mu_1$ and $\mu_2$. For the case of  $( \mu_1 \geq  \mu_2 \geq  \mu_0 )$, point $T$ in the boundary region maximizes the capacity and $\mu_1 B_1 +  \mu_2 B_2 + \mu_0 B_0  $ becomes $\mu_0  \sum_{n = 1}^{N}  C \big( P^s_{0,2}(n)  \big) + (\mu_2 \, - \, \mu_0)  \sum_{n = 1}^{N} C \big( P^s_{1,2}(n) \big) + (\mu_1 \, - \, \mu_2)    \sum_{n = 1}^{N} C \big(  P_1(n) \big) $. For the case of, $(\mu_1 \geq  \mu_0 \geq  \mu_2)$, $P_2$ is equal to zero until $T_f$ where the boundary point $S$ maximizes the capacity and $\mu_1 B_1 +  \mu_2 B_2 + \mu_0 B_0  $ becomes  $\mu_0 \,\sum_{n = 1}^{N}  C \big( P^s_{0,1}(n)  \big) + (\mu_1 \, - \, \mu_0) \sum_{n = 1}^{N} C \big( P_1(n) \big)$. And finally, for $(\mu_0 \geq  \mbox{max}(\mu_1, \mu_2) )$, in the optimum solution, $P_1$ and $P_2$ will be zero and all the power will be consumed for the common data with $   \mu_0  \sum_{n = 1}^{N} C \big( P_0(n)   \big)  $ in the objective function. Now, the solution for three different cases of rewarding values $( \mu_1,  \mu_2 , \mu_0 )$ are analyzed.

\subsection{Optimum Scheduling Solution for Capacity Regions }
\subsubsection{ $( \mu_1 \geq  \mu_2 \geq  \mu_0 )$}
This case corresponds to the boundary point $T$, and taking the derivative of (\ref{eqoptprob2}) with respect to  $P_1(n), \, P_2(n), \, P_0(n), \, \rho(n)$ and equalizing to zero give the following KKT equalities for $n \in \left[1,N \right]$,  
\begin{eqnarray}
   \label{eq50} 
  P_{1,n}&\mbox{:}&   \frac{\mu_0}{1 + \sum_{k = 0}^{2} P_k^{'}(n) } + \frac{\mu_2 - \mu_0}{1 +  P_1^{'}(n) + P_2^{'}(n) }  
   + \frac{\mu_1 - \mu_2}{1 +  P_1^{'}(n) } = \lambda_{1,n}^{p} - \lambda_{5,n} \, \, \, \, \, \,  \\
     \label{eq51} 
    P_{2,n}&\mbox{:}&   \frac{\mu_0}{1 +  \sum_{k = 0}^{2} P_k^{'}(n) } + \frac{\mu_2 - \mu_0}{1 +  P_1^{'}(n) + P_2^{'}(n) }   
   = \lambda_{2,n}^{p} - \lambda_{6,n}   \\
       \label{eq52} 
    P_{0,n}&\mbox{:}&   \frac{\mu_0}{1 +  \sum_{k = 0}^{2} P_k^{'}(n) }   =  
     \lambda_{1,n}^{p} \rho^2(n) + \lambda_{2,n}^{p} \big(1 - \rho(n)\big)^2 - \lambda_{7,n}    
 \end{eqnarray}
for $\rho(n)$ in (\ref{eq54}) where  $P_{k}^{'} = P_k \, / \, A$ for $k \in [0,2]$ and $\lambda^{'} = \lambda \, A \, \mbox{log}_{\mbox{e}} 2  \, / \, W_{Tot}$ is replaced with $\lambda$ in order to simplify the notation without changing the final solution for $P_{i}^{'}$. For $P_0^{'}(n) > 0$, three different regions of  $\rho(n) $ can be observed, i.e., $0 < \rho(n) < 1$, $\rho(n) = 0$ and $\rho(n) =1$. For the first case, i.e., $0 < \rho(n) < 1$, $\rho(n)$ is obtained in (\ref{eq54}). If the resulting expression is inserted into (\ref{eq52}),  the following can be obtained,
\begin{equation}
       \label{eq55} 
 P_{0,n} \mbox{:  }    \frac{\mu_0}{1 +  P_0^{'}(n) + P_1^{'}(n) + P_2^{'}(n) }   = g ( \lambda_{1,n}^{p},  \lambda_{2,n}^{p})  
\end{equation}
where  $g ( \lambda_{1,n}^{p},  \lambda_{2,n}^{p})$ is defined as
\begin{equation}
       \label{eq55a} 
   g ( \lambda_{1,n}^{p},  \lambda_{2,n}^{p})     \triangleq \frac{  \lambda_{1,n}^{p}  \lambda_{2,n}^{p}}{ \left(  \lambda_{1,n}^{p} + \lambda_{2,n}^{p}  \right) }   
\end{equation}
Furthermore, the overall power consumed in the time interval $n$ should satisfy the total power equalities as $\overline{P_1}(n) = P_1(n) + \rho^2(n) P_0(n)$ and $\overline{P_2}(n) = P_2(n) + \big( 1 -  \rho(n)\big)^2 P_0(n)$. If the expressions of $\rho(n)$ and $P_0^{'}(n)$ in (\ref{eq54}) and (\ref{eq55}), respectively, in terms of $\lambda_{1,n}^{p}$ and $\lambda_{2,n}^{p}$ are inserted to these power equalities, the resulting equations are obtained for $\lambda_{1,n}^{p}$ and $\lambda_{2,n}^{p}$ in terms of $\big(P_0(n), P_1(n), P_2(n), \overline{P_1}(n), \overline{P_2}(n)\big)$ as the following,
\begin{equation}
       \label{eqlambdaopt} 
\lambda_{k,n}^{p}    = \frac{\mu_0}{1 +  \sum_{k = 0}^{2} P_k^{'}(n)}\frac{1 + \chi_k(n)}{\chi_k(n)}; \, \, \, \, \chi_k(n) = \sqrt{ \frac{\overline{P_k}(n) - P_k(n)}{\overline{P_{3-k}}(n) - P_{3-k}(n)}}   
\end{equation}
where $k \in [1,2]$ and  $n \in [1,N]$.
For the second and third cases of $\rho(n)$, i.e., $\rho(n) = 0$ and $\rho(n) =1$, the following is obtained from (\ref{eq52}),
\begin{equation}
P_{0,n} \mbox{:  }    \frac{\mu_0}{1 +  P_0^{'}(n) + P_1^{'}(n) + P_2^{'}(n) }   = \begin{cases}   \lambda_{2,n}^{p}     &\mbox{if }\rho(n) = 0\\ 
 \lambda_{1,n}^{p}     & \mbox{if } \rho(n) = 1. \end{cases}  
\end{equation}  
Therefore,    $\lambda_{1,n}^{p}$ and $\lambda_{2,n}^{p}$ are expressed in terms of $\big(P_0(n), P_1(n), P_2(n), \overline{P_1}(n), \overline{P_2}(n)\big)$. For the other cases, the equalities in (\ref{eq50}-\ref{eq52}) can be used to extract the values of $\lambda_{1,n}^{p}$ and $\lambda_{2,n}^{p}$ whenever the respective Lagrange multipliers, i.e., $\lambda_5(n)$ or $\lambda_6(n)$, are zero corresponding to $P_1(n) > 0$ or $P_2(n) > 0$, respectively. Now, the water-filling algorithm for the defined solution is provided finding the global optimum solution in an efficient way.

\begin{thm} 
\label{lemma3} 
The optimization regarding the equalities in (\ref{eq50}-\ref{eq52}) is achieved by defining and equalizing the following water levels  $WL_k(n)$ for each time interval $n \in [1,N]$ such that $WL_k(n) \leq WL_k(n+1)$ is satisfied for $ k \in [1,2,3]$ and the water levels are defined   as follows
\begin{eqnarray}
WL_1(n) & =& \big( \frac{\mu_0}{1+ \sum_{k=1}^{3}P_k^{'}(n)} + \frac{\mu_2 - \mu_0}{1+ \sum_{k=1}^{2}P_k^{'}(n)} +   \frac{\mu_1 - \mu_2}{1+ P_1^{'}(n) }  \big)^{-1} \\
WL_2(n) & =& \big( \frac{\mu_0}{1+ P_1^{'}(n) + P_2^{'}(n) + P_0^{'}(n)} + \frac{\mu_2 - \mu_0}{1+ P_1^{'}(n) + P_2^{'}(n)}  \big)^{-1}    \, \, \, \, \, \, \, \, \, \\
WL_3(n) & =& \big( \frac{\mu_0}{1+ P_1^{'}(n) + P_2^{'}(n) + P_0^{'}(n)}     \big)^{-1}
 \end{eqnarray}
and $WL_4(n)$, $WL_5(n)$ obtained by (\ref{eq50} - \ref{eq52}), (\ref{eqlambdaopt}) are defined in Table \ref{taboptlambda}  where $+$ denotes $ > 0$ and the water levels satisfy the following,
\begin{equation}
\begin{cases}   
WL_4(n) = WL_4(n+1)    &\mbox{if }\lambda_{1,n}^{p} = \lambda_{1,n+1}^{p}\\ 
WL_5(n) = WL_5(n+1)   &\mbox{if }\lambda_{2,n}^{p} = \lambda_{2,n+1}^{p}\\ 
WL_{i}(n) = WL_i(n+1), \, i =[4,5]   & \mbox{if }\lambda_{1,n}^{p} = \lambda_{1,n+1}^{p} \,  \& \, \lambda_{2,n}^{p} = \lambda_{2,n+1}^{p}    
\end{cases}  
\end{equation}  
\end{thm} 
   
\newcolumntype{M}[1]{>{\centering\arraybackslash}m{#1}}
\renewcommand{\arraystretch}{1.7}
\begin{table*}[t!]
\caption{The water levels  for the time interval $n$ for $(\mu_1  \geq  \mu_2  \geq \mu_0 )$}
\centering
\footnotesize
\begin{tabular}{| M{3.9cm}|M{3.9cm}|p{0.6cm}|p{0.6cm}|p{0.6cm}|p{0.6cm}|p{0.6cm}|}  
\hline
$\big( WL_4(n) \big)^{-1} $   & $\big( WL_5(n) \big)^{-1}$  & $\overline{P_1}$  & $\overline{P_2}$  & $P_0$  & $P_1$  & $P_2$   \\  \hline 
         - &   \vspace{-0.5em}$\frac{\mu_2}{1 +  P_2^{'}}$     & $ 0$  & $+$  & $ 0$  & $ 0$  & $+$   \\\hline 
             \vspace{-0.5em} $\frac{\mu_1}{1 +  P_1^{'}}$     &     -  & $+$  & $0$  & $ 0$  & $+$  & $0$   \vspace{-0.5em}\\\hline 
                        \vspace{-0.5em}  $\frac{\mu_2}{1 + P_1^{'} + P_2^{'}}  + \frac{\mu_1 - \mu_2}{1 +  P_1^{'}} $        &        \vspace{-0.5em}   $\frac{\mu_2}{1 + P_1^{'} + P_2^{'}}$ & $ +$  & $+$  & $ 0$  & $+$  & $+$   \\\hline 
 \vspace{-0.5em} $\frac{\mu_0}{1 +    P_0^{'} +   P_1^{'} +   P_2^{'}} \frac{1 + \chi_1}{\chi_1} $   &  \vspace{-0.5em} $\frac{\mu_0}{1 +   P_0^{'} +   P_1^{'} +   P_2^{'}} \frac{1 + \chi_2}{\chi_2} $  & $+$  & $+$  & $+$  & $0$ &   $  0$  \\  \hline 
    \vspace{-0.5em}   $\frac{\mu_0}{1 +   P_0^{'} +   P_1^{'} +   P_2^{'}} \frac{1 + \chi_1}{\chi_1} $          &   \vspace{-0.5em}  $\frac{\mu_0}{1 + P_0^{'} + P_2^{'}}  + \frac{\mu_2 - \mu_0}{1 +  P_2^{'}} $     & $ +$  & $+$  & $ +$  & $ 0$  & $+$   \\\hline 
    \vspace{-0.5em}          $\frac{\mu_0}{1 + P_0^{'} + P_1^{'}}  + \frac{\mu_1 - \mu_0}{1 +  P_1^{'}} $    &    \vspace{-0.5em} $\frac{\mu_0}{1 +   P_0^{'} +   P_1^{'} +   P_2^{'}} \frac{1 + \chi_2}{\chi_2} $    & \vspace{-0.5em} $ +$  & $+$  & $ +$  & $ +$  & $0$   \\\hline 
    \vspace{-0.5em} $\frac{\mu_0}{1 +  P_0^{'} +   P_1^{'} +   P_2^{'}} \frac{1 + \chi_1}{\chi_1} $                         &    \vspace{-0.5em}   $\frac{\mu_0}{1 +   P_0^{'} +   P_1^{'} +   P_2^{'}} \frac{1 + \chi_2}{\chi_2} $     & $+$  & $ +$  & $+$  & $+$  & $+$\\  \hline 
\end{tabular}
\label{taboptlambda}
\end{table*}
\renewcommand{\arraystretch}{1}
The proof of Theorem \ref{lemma3} is given in Appendix \ref{app0}.

The constraints and the components of the objective functions for the other 2 cases of $(\mu_1, \mu_2, \mu_0 )$, i.e., $\mu_1 \geq  \mu_0 \geq  \mu_2 $ and $\mu_0 \geq \mbox{max}(\mu_1, \mu_2)$,  are found with a similar approach to the case 1. Therefore, in the following, only the KKT conditions are presented for thee cases without the detailed proofs.  
 
\subsubsection{ $\mu_1 \geq  \mu_0 \geq  \mu_2 $}  

This case corresponds to the boundary point $S$, and similar to the approach performed for $\mu_1 \geq  \mu_2 \geq  \mu_0 $, taking the derivative of (\ref{eqoptprob2}) with respect to  $P_1(n),  P_0(n), \, \rho(n)$ gives the following for $n \in \left[1,N \right]$,
\begin{eqnarray}
       \label{eq56} 
  P_{1,n}&\mbox{:}&   \frac{\mu_0}{1 +  P_0^{'}(n) + P_1^{'}(n)   } + \frac{\mu_1 - \mu_0}{1 +  P_1^{'}(n)   }   
    = \lambda_{1,n}^{p} - \lambda_{5,n}   \\
       \label{eq57} 
    P_{0,n}&\mbox{:}&   \frac{\mu_0}{1 +  P_0^{'}(n) + P_1^{'}(n)   }   =  
    \lambda_{1,n}^{p} \rho^2(n) + \lambda_{2,n}^{p} \big(1 - \rho(n)\big)^2 - \lambda_{7,n}    \, \, \, \,  \, \, \\
       \label{eq58} 
  \rho(n)&\mbox{:}&   P_0(n)\left[ \rho(n) \left(2  \lambda_{1,n}^{p} + 2  \lambda_{2,n}^{p} \right)  \right]  
 =  P_0(n) \left[ 2  \lambda_{2,n}^{p}   \right] + \lambda_{3,n} - \lambda_{4,n}   
 \end{eqnarray}
If $0 < \rho(n) < 1$, (\ref{eq54}) is satisfied and putting into (\ref{eq57}), the equation becomes as the following,
\begin{equation}
       \label{eq60} 
 P_{0,n} \mbox{:  }    \frac{\mu_0}{1 +  P_0^{'}(n) + P_1^{'}(n)  }   = g ( \lambda_{1,n}^{p},  \lambda_{2,n}^{p} ) - \lambda_{7,n}    
\end{equation} 

It can be observed that $P_2(n)$ becomes always zero and the water level $WL_2(n)$ is absent compared with the first case. Then, it can be proved that $WL_1(n) = \big( \frac{\mu_0}{1+  P_0^{'}(n)+  P_1^{'}(n)} + \frac{\mu_1 - \mu_0}{1+ P_1^{'}(n)} \big)^{-1}$ and    
$WL_3(n)  =  \big( \frac{\mu_0}{1+ P_1^{'}(n)  + P_0^{'}(n)} \big)^{-1}$. Therefore, the optimization algorithm will check for the inequalities $WL_{i}(n) \geq WL_{i}(n+1)$ for $i = 1$ and $i = 3$, and equalities $WL_{i}(n)= WL_{i}(n+1)$ for $i = 4$ or $i = 5$. Furthermore, the table regarding $WL_{4,5}(n)$ is modified as shown in Table \ref{taboptlambda2}.

\renewcommand{\arraystretch}{1.7}
\begin{table*}[t!]
\caption{The water levels for the time interval $n$ for $(\mu_1  \geq  \mu_0 \geq \mu_2 )$}
\begin{center}
\footnotesize
\begin{tabular}{| M{3.9cm}|M{3.9cm}|M{0.6cm}|M{0.6cm}|M{0.6cm}|M{0.6cm}|M{0.6cm}|}
\hline
$\big( WL_4 \big)^{-1} $   & $\big(  WL_5  \big)^{-1}$  & $\overline{P_1}$  & $\overline{P_2}$  & $P_0$  & $P_1$  & $P_2$    \\\hline 
 - &   \vspace{-0.5em}$\frac{\mu_2}{1 +  P_2^{'}}$     & $ 0$  & $+$  & $ 0$  & $ 0$  & $+$   \\\hline 
            \vspace{-0.5em}  $\frac{\mu_1}{1 +  P_1^{'}}$     &     -  & $+$  & $0$  & $ 0$  & $+$     & $ 0$ \\\hline 
\vspace{-0.5em}  $\frac{\mu_0}{1 +  P_0^{'} +  P_1^{'}} \frac{1 + \chi_1}{\chi_1}  $   &  \vspace{-0.5em} $\frac{\mu_0}{1 + P_0^{'} +  P_1^{'}} \frac{1 + \chi_2}{\chi_2}  $  & $+$  & $+$  & $+$  & $0$   & $ 0$ \\\hline%
 \vspace{-0.5em}   $\frac{\mu_0}{1 + P_0^{'} + P_1^{'}}  + \frac{\mu_1 - \mu_0}{1 +  P_1^{'}} $    &    \vspace{-0.5em} $\frac{\mu_0}{1 + P_0^{'} +  P_1^{'}} \frac{1 + \chi_2}{\chi_2}  $    & $ +$  & $+$  & $ +$  & $ +$  & $ 0$  \\\hline 
\end{tabular}
\end{center}
\label{taboptlambda2}
\end{table*}
\renewcommand{\arraystretch}{1}
  
\subsubsection{ $\mu_0 \geq  \mbox{max}(\mu_1, \mu_2) $}   

This case corresponds to the boundary point $Q$ and taking the derivative of (\ref{eqoptprob2}) with respect to  $ P_0(n), \, \rho(n)$ and equalizing to zero give  the following KKT conditions for $n \in \left[1,N \right]$, 
\begin{equation}
       \label{eq61} 
     P_{0,n} \mbox{:  }      \frac{\mu_0}{1 +  P_0^{'}(n)   } =  \lambda_{1,n}^{p} \rho^2(n) + \lambda_{2,n}^{p} \big(1 - \rho(n)\big)^2   
 \end{equation}
Furthermore, (\ref{eq54}) is satisfied and putting into (\ref{eq57}), the equation becomes as the following,
\begin{equation}
       \label{eq62} 
  \frac{\mu_0}{1 +  P_0^{'}(n)   } = g ( \lambda_{1,n}^{p},  \lambda_{2,n}^{p} )   
 \end{equation}
 
It is observed that $P_1(n) = P_2(n) = 0$ for $n \in [1,N]$, the comparison for $WL_{1,2}(n)$ is removed and only $WL_3(n) = \big( \frac{\mu_0}{1 +  P_0^{'}(n)} \big)^{-1}$ is compared between neighbouring time intervals. The table for $WL_{4,5}(n)$ is modified as shown in Table \ref{taboptlambda3}.

\renewcommand{\arraystretch}{1.7}
\begin{table*}[t!]
\caption{The water levels  for the time interval $n$ for $(\mu_0  \geq \mbox{max}(\mu_1, \mu_2) )$}
\begin{center}
\footnotesize
\begin{tabular}{| M{3.9cm}|M{3.9cm}|M{0.6cm}|M{0.6cm}|M{0.6cm}|M{0.6cm}|M{0.6cm}|}
\hline
$\big(  WL_4 \big)^{-1}$   & $\big(  WL_5 \big)^{-1}$  & $\overline{P_1}$  & $\overline{P_2}$  & $P_0$ & $P_1$  & $P_2$       \\\hline 
 - &   \vspace{-0.5em}$\frac{\mu_0}{1 +  P_0^{'}}$     & $ 0$  & $+$  & $ +$  & $ 0$  & $0$   \\\hline 
            \vspace{-0.5em}  $\frac{\mu_0}{1 +  P_0^{'}}$     &     -  & $+$  & $0$  & $+$  & $ 0$     & $ 0$ \\\hline 
 \vspace{-0.5em}  $\frac{\mu_0}{1 +  P_0^{'}} \frac{1 + \chi_1}{\chi_1} $   &  \vspace{-0.5em}  $\frac{\mu_0}{1 + P_0^{'} } \frac{1 + \chi_2}{\chi_2} $  & $+$  & $+$  & $+$  & $0$  & $0$   \\\hline 
\end{tabular}
\end{center}
\label{taboptlambda3}
\end{table*}
\renewcommand{\arraystretch}{1}

\subsection{Iterative Water-Filling Algorithm}   
It is difficult  to implement the solution defined in Theorem \ref{lemma3} in a water-filling algorithm since at each time interval either of the nodes or both of the nodes can transfer energy to the neighbouring next time interval by looking only at the neighbouring time intervals with possibly leading to sub-optimum solution. Therefore, an iterative water-filling algorithm is defined similar to \cite{b12} in a way realizing power scheduling node by node iteratively while fixing the energy levels in the other node. Then, the following lemma is proved for $\mu_1 \geq \mu_2 \geq \mu_0$ and it can be easily extended to the other 2 cases of $\mu_1 \geq \mu_0 \geq \mu_2$ and $ \mu_0 \geq \mbox{max}(\mu_1, \mu_2)$.
\begin{thm} 
\label{lemma4} 
The optimization regarding the equalities in (\ref{eq50}-\ref{eq52}) can be achieved with an iterative backward water-filling algorithm given in Algorithm \ref{alg1} by only satisfying  $WL_4(n) = WL_4(n+1)$ for the power scheduling of the 1st node where the power levels in 2nd is fixed, and $WL_5(n) = WL_5(n+1)$   for the power scheduling of the 2nd node where the power levels in 1st is fixed for $n \in [1,N-1]$ whenever power transfer from the time interval $n$ to $n+1$ occurs.
\end{thm} 
The proof is given in  Appendix \ref{app0}.
 
The iterative water-filling algorithm given in Algorithm \ref{alg1} looks at the difference between water-levels and using an iterative weighted search algorithm, water-levels are optimized. The algorithm firstly optimizes the power levels in all time intervals based on the available amount of power in each node and using the iterative algorithm based on the regions and solutions in Table \ref{tab_app1} in Appendix \ref{app1}. Then, the power levels in the nodes are iteratively  optimized by equalizing $WL_4(n) = WL_4(n+1)$ and $WL_5(n) = WL_5(n+1)$  for the power scheduling of the 1st and 2nd node, respectively, while fixing the power levels in the other node fixed. For the neighbouring time intervals $n$ and $n+1$ if the equalities cannot be achieved by the power transfer then no action is taken for that specific time interval $n$.  
\begin{algorithm}[!t]
\scriptsize
\caption{Iterative Water-filling Algorithm for Gaussian MAC with Common Data}
\label{alg1}
\begin{algorithmic}
\STATE Initialize $WL_{i}(n)$ for $i \in [4,5]$  and $n \in [1,N]$ by finding the optimum $P_0(n)$, $P_1(n)$, $P_2(n)$ with the solutions in Table \ref{tab_app1}.
\FORALL{$i = 4$ to $5$}
\WHILE{$WL_{i}(n) \neq  WL_{i}(n+1)$  }
\FORALL{time intervals $n = N-1$ to $1$}
\IF{$i = 4$}
\STATE Fix 2nd node fixed and satisfy $WL_{4}(n) =  WL_{4}(n+1)$ for the 1st node.
\STATE Update $P_{k}(n)$, $P_{k}(n+1)$, $WL_{4}(n)$, $WL_{4}(n+1)$  for $k \in [0,1,2]$.
\ENDIF
\IF{$i = 5$}
\STATE Fix the 1st node fixed and satisfy $WL_{5}(n) = WL_{5}(n+1)$ for the 2nd node.
\STATE Update $P_{k}(n)$, $P_{k}(n+1)$, $WL_{5}(n)$, $WL_{5}(n+1)$  for $k \in [0,1,2]$.
\ENDIF
\ENDFOR
\ENDWHILE
\ENDFOR
\end{algorithmic}
\end{algorithm}
  
\section{Numerical Simulation Results}
\label{simulr1}

The proposed algorithm is simulated for the case that the 1st user and 2nd user harvest energies $E_1 = \left[ 3 \, 6 \, 10 \right]$ and $E_2 = \left[ 4 \, 11 \, 6 \right]$ mJ at the time instants $\left[ 0 \, 2 \, 6 \right]  $ and $\left[ 0 \, 5 \, 8 \right] $  seconds, respectively, $N = 5$ time intervals are considered until $T_f = 11$ seconds, bandwidth $W_{Tot} = 1$ Mhz,  noise spectral density $N_0 = 10^{-19}$ W/Hz and path loss $h = 10^{-11}$ are considered for generic analyses used in the literature \cite{b11}.  In the following figures, $E_{i}^c$ and $E_{i}^r$ denote consumed and harvested energy,respectively, for the nodes $i \in [1,2]$. 

\begin{figure}[ht!]
\begin{center}
\includegraphics[width=3in]{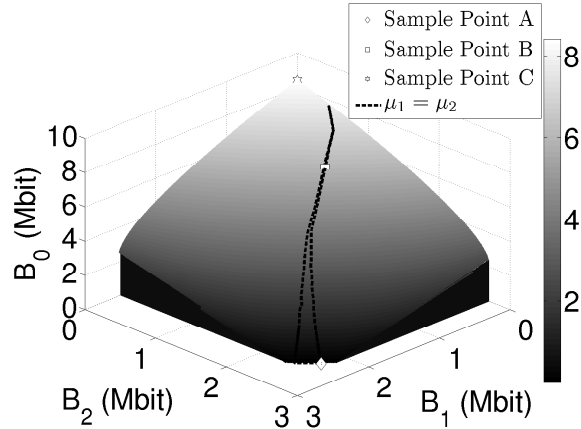} \\ 
\caption{The capacity boundary surface formed by using a large set of $(\mu_1, \mu_2, \mu_0)$. }
\label{figure_ns_1}
\end{center}
\end{figure}

The departure region boundary surface formed from the scattered points of the sample points is shown in Fig. \ref{figure_ns_1} which resembles the capacity boundary surface for the  single time interval shown in Fig. \ref{fig2}. It can be observed that there is a linear region on the capacity boundary surface between the points obtained with $\mu_1 = \mu_2$. This region corresponds to sampling boundary points between $T$ and $U$ as shown in Fig. \ref{fig2}. Furthermore, 3 sample points labelled with the labels $A$, $B$ and $C$ on the $\mu_1 = \mu_2$ curve are taken which corresponds to sampling the point $U$ at different $\mu_0$. These points are used to compare the effect of the amount of $B_0$ on the optimum power scheduling scenario. These points correspond to $B_0 = 0$ for the case when no common data is transmitted similarly to the study in \cite{b11}, $B_0 = 5.41$ Mbit where common data and the distinct data of each user are transmitted and the case for the maximum amount of common data of $B_0 = 8.41$ Mbit corresponding to also the maximum amount of total data rate.  

\begin{figure}[t!]
\begin{center}
\includegraphics[width=3in]{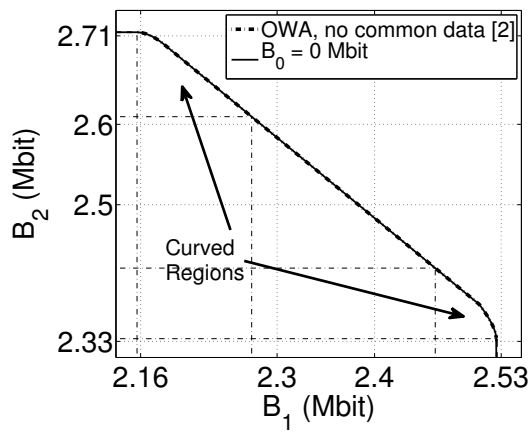} \\ 
\caption{The comparison between the optimum water-filling algrithm (OWA) for MAC scheduling framework in  \cite{b11}  and the proposed optimum water-filling algorithm for $B_0 = 0$. }
\label{figure_ns_2}
\end{center}
\end{figure}

Firstly, the comparison between the defined optimization framework and the optimum backward iterative water-filling algorithm denoted as OWA in \cite{b11} is shown in Fig. \ref{figure_ns_2}. It is observed that the proposed solution gives the same result with \cite{b11} when no common data is transmitted, e.g., $B_0 = 0$. Furthermore, the linear region between the points $T$ and $U$ can be observed.  The constant $B_0$ contours are given in Fig. \ref{figure_ns_3}. While $B_0$ increases, the amount of power used for $B_1$ and $B_2$ decreases and the boundary shows a curved behaviour with a decreasing amount of linear region as shown in Fig. \ref{figure_ns_3}.
\begin{figure}[t!]
\begin{center}
\includegraphics[width=3in]{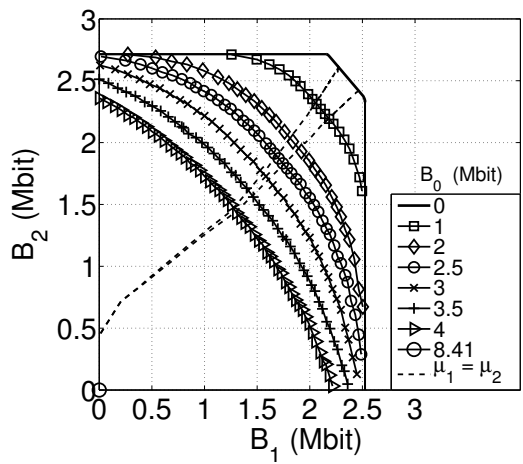} \\ 
\caption{Constant $B_0$ contours for the optimum scheduling algorithm.}
\label{figure_ns_3}
\end{center}
\end{figure}
\vspace{0.2in}
\begin{figure}[t!]
\begin{center}
\includegraphics[width=3in]{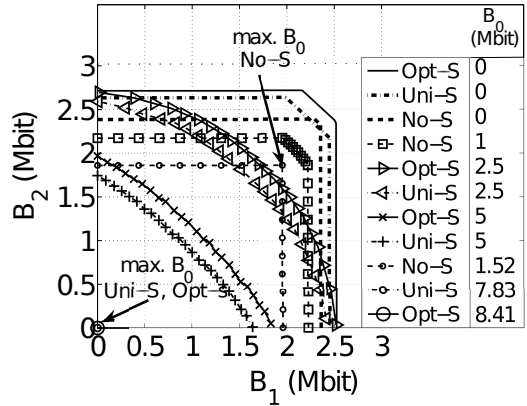} \\ 
\caption{Comparison between the no scheduling (No-S), uniform power distribution scheduling (Uni-S) and the optimum scheduling algorithm (Opt-S) for constant $B_0$ contours.}
\label{figure_ns_4}
\end{center}
\end{figure}

The optimum scheduling algorithm is denoted as Opt-S and compared with two basic algorithms denoted as Uni-s and No-S which represent no scheduling and the scheduling algorithm distributing the power uniformly until $T_f$, respectively. In No-S algorithm, at each time interval, the power levels are optimized to maximize the data throughput only in that time interval by using the solutions in Table \ref{tab_app1} in  \ref{app1}. In Uni-S algorithm, the harvested power is distributed uniformly for the next-coming time intervals until $T_f$. As shown in Fig. \ref{figure_ns_4}, No-S algorithm gives a maximum amount of common data rate of approximately $B_0 = 1.52$ Mbit with non-zero $B_1$ and $B_2$ since in some time intervals the harvested power is not utilized for $B_0$ due to the reason that only one of the nodes harvest energy and the optimal solution does not allow allocating power to common data. The Uni-S and Opt-S algorithms perform better compared with the No-S algorithm with larger $B_0$ and a larger capacity boundary volume. Opt-S algorithm is better compared with Uni-S algorithm as shown in Fig. \ref{figure_ns_4} such that a larger $B_1$ and $B_2$ curve is obtained for constant $B_0$ contours and the maximum amount of $B_0$ is bigger.

\begin{figure}[t!]
\begin{center}
\includegraphics[width=3in]{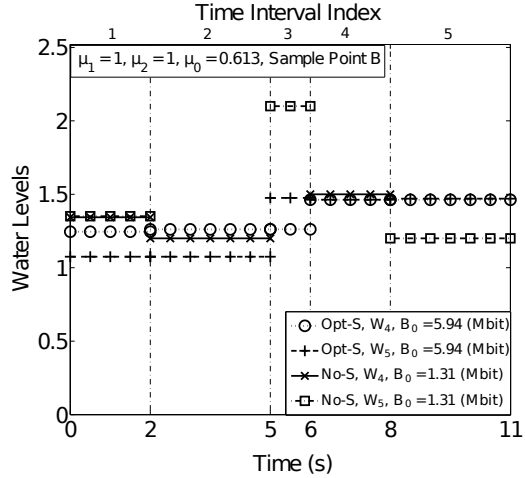} \\ 
\caption{The water-filling levels for No-S and Opt-S algorithms for the Sample Point B corresponding to $\mu_1=  \mu_2=1,   \mu_0 = 0.613$.}
\label{figure_ns_5}
\end{center}
\end{figure}

The water-filling and the power scheduling profiles of the nodes are shown and compared in Figs. \ref{figure_ns_5} - \ref{figure_ns_8} for the sample points $A$, $B$ and $C$ described in Fig. \ref{figure_ns_1}. As shown in Fig. \ref{figure_ns_5}, when there is no power scheduling, the water levels are not in equilibrium resulting a lower data rate for the common data $B_0$. On the other hand, optimization algorithm leads to the equilibrium of the water levels such that $W_4(n) = W_4(n+1)$ for $n = 2, 4$ and $W_5(n) = W_5(n+1)$ for $n = 1, 3, 4$ where previously harvested power is transferred to 3rd and 5th time intervals for the first node and to 2nd, 4th and 5th intervals for the second node.

\begin{figure}[t!]
\begin{center}
\includegraphics[width=3in]{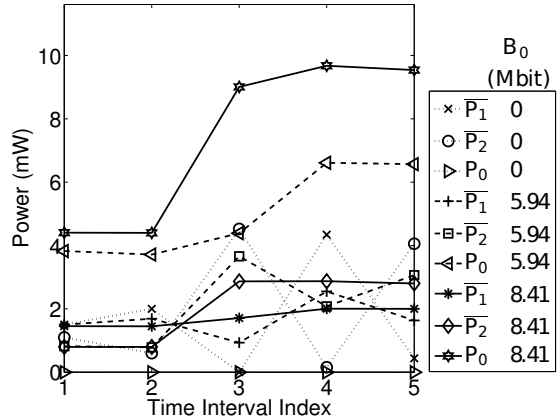} \\ 
\caption{ Consumed power levels in each time interval for the sample points $A$, $B$ and $C$ for varying $B_0$.}
\label{figure_ns_6}
\end{center}
\end{figure}

It can be observed that as $B_0$ increases, the power levels consumed by each node in different time intervals are more homogenized in order to maximize $B_0$ as shown in Fig. \ref{figure_ns_6} since for a given $\overline{P_1} + \overline{P_2}$ more amount of $P_0$ is obtained whenever $\overline{P_1}$ and $\overline{P_2}$ are closer to each other. The same situation is observed in  Fig. \ref{figure_ns_7} where more homogeneous and increasing power consumption is realized for each node as $B_0$ increases.

\begin{figure}[t!]
\begin{center}
\includegraphics[width=3in]{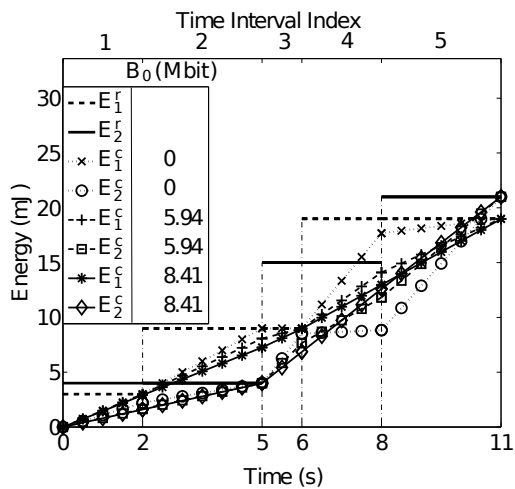} \\ 
\caption{Consumed energy levels for each node for varying $B_0$ at  the sample points A, B an C.  }
\label{figure_ns_7}
\end{center}
\end{figure}

Furthermore, as shown in Fig. \ref{figure_ns_8}, for the case when no common data is transmitted and $\mu_1 = \mu_2$, the optimal power scheduling  optimizes as if there is a single node with the total power $E_1^r + E_2^r$ \cite{b7, b11} where the total consumed power monotonically increases and uses all the harvested energy until the transmission rate changes. However, as $B_0$ increases, some of the harvested power is saved for future use in a way to maximize $B_0$ for a given $B_1$, $B_2$ data rate. For example, at $t = 5$ (sec)  not all the harvested energy is used although there is an increase in the data rate for the common data at the sample boundary points B an C as shown in Figs. \ref{figure_ns_6} and \ref{figure_ns_8}.

\begin{figure}[t!]
\begin{center}
\includegraphics[width=3in]{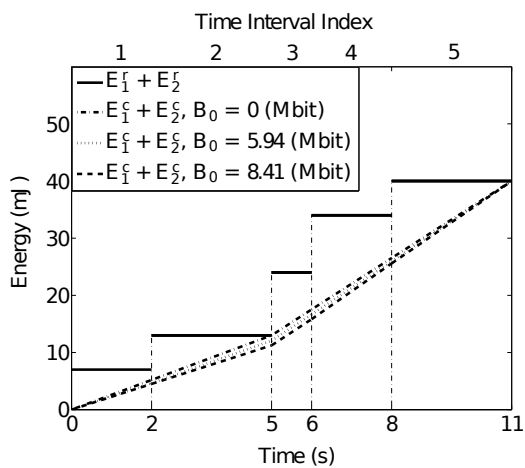} \\ 
\caption{Consumed total energy levels for varying $B_0$ at  the sample points A, B an C.  }
\label{figure_ns_8}
\end{center}
\end{figure}

As a result, the optimum solution defined in this article extends the previously defined optimum solution for Gaussian MAC without common data and performs better than no power scheduling and uniform power scheduling cases. Furthermore, although water-filling levels are more complicated due to the complexity of the optimum departure region for Gaussian MAC with common data, very efficient and simple water-filling algorithm is defined for the optimum solution.

\section{Future Work and Open Issues}
\label{fw}

There is a multitude of assumptions defined in the article which leads to a set of future work topics. First of all,   the optimum power allocation solution for $m$-user Gaussian MAC with common data in a single time interval, e.g. the study in \cite{b20}, can be utilized for an optimum packet scheduling policy with energy harvesting transmitters and various levels common data sharing mechanisms between $m$ nodes.

Furthermore, another open issue is to combine the time-varying fading effects and stochastic nature of energy harvesting and packet reception in a way to obtain the optimum online solution for Gaussian MAC with common data. Furthermore, the online extension can be improved more by adapting the scheduling algorithm to the level of the correlation among the transmitters, e.g., in a temporally or spatially correlated WSN.

Moreover, the utilization of common data beam-forming to transfer not only data but also energy, e.g., in a sensor networking architecture with limited amount of powers, is an alternative open issue to be analyzed. Time varying fading effects can be considered to adapt the transmission policy based on the current channel states. Therefore, an optimization framework can be defined deciding on the trade off between the  amount of common data to transfer more energy to the next node and the amount of data for the independent messages.  
 
\section{Conclusion}
\label{conc}

In this paper, optimum offline packet scheduling policy is developed for 2-user Gaussian MAC with common data and energy harvesting transmitters.  The optimum solution is derived by using KKT conditions. The proposed solution is implemented with an efficient iterative water-filling algorithm. The proposed optimum solution algorithm is numerically simulated to find the optimum departure region.  The optimum solution is compared with the solutions defined for Gaussian MAC without common data, no scheduling and uniform power scheduling solutions with the result of being the best among them and extending the previous solutions for Gaussian MAC without common data. Finally, a set of open issues and future work studies are defined which includes considering stochastic nature of the system leading to online solution, time-varying fading effects, both energy and data transfer possibility and the optimum offline solution for m-user Gaussian MAC with common data.

\begin{appendices}

\section{}
\label{app0}

\begin{proof}[Proof of Lemma \ref{lemma2}]
We can use a similar approach to \cite{b2} and extend the proof to include multiple time intervals. In the following, $\boldsymbol{B}$ denotes $(B_0, B_1, B_2)$, $\boldsymbol{P} = \left[P(1) \, P(2) \, \hdots \, P(N) \right]$ and $\boldsymbol{\rho} = \left[\rho(1) \, \rho(2) \, \hdots \, \rho(N) \right]$. We need to show that for any $0 < \theta < 1$, $\Big( \theta \boldsymbol{B^a}  + (1 - \theta)\boldsymbol{B^b}, $ $\theta \boldsymbol{\widetilde{E_{1, a}}} + (1 - \theta)\boldsymbol{\widetilde{E_{1, b}}},$ $  \theta \boldsymbol{\widetilde{E_{2, a}}} + (1 - \theta)\boldsymbol{\widetilde{E_{2, b}}}   \Big)$ is in  $\mathbb{P}$.  Both $\boldsymbol{B^a}$ and $\boldsymbol{B^b}$ satisfy that  $ \boldsymbol{B^{s}} 
\in B(\boldsymbol{P_1^{s}}, \boldsymbol{P_2^{s}}, \boldsymbol{P_0^{s}}, \boldsymbol{\rho^{s}})$ for some $ \boldsymbol{P_1^{s}}, \boldsymbol{P_2^{s}}, \boldsymbol{P_0^{s}}, \boldsymbol{\rho^{s}}$ such that the following conditions hold
\begin{eqnarray}
\label{l2_eq1}
\sum_{n = 1}^{i} \left( P_{1}^{s}(n) + \big(\rho^{s}(n)\big)^2 P_0^{s}(n) \right) L(n) \, \leq \, \widetilde{E_{1,s}}(i)    && \\
\label{l2_eq2}
\sum_{n = 1}^{i} \left(  P_2^{s}(n) + \big( 1 -  \rho^{s}(n)\big)^2 P_0^{s}(n) \right) L(n) \, \leq \, \widetilde{E_{2, s}}(i)  &&  
\end{eqnarray}
for $i \in [1, N]$ and $s = a$ or $s = b$. Let us define  $\boldsymbol{P_{k}} \equiv \theta \boldsymbol{P_{k}^{a}} + ( 1 - \theta ) \boldsymbol{P_{k}^{b}}$ for $k \in [0,2]$ and the following for $n  \in [1, N]$
\begin{eqnarray}
\label{l2_eq3}
\rho_1(n) &=& \sqrt{\frac{\theta  P_0^{a}(n) \big( \rho^a(n)\big)^2 + (1 - \theta)  P_0^{b}(n) \big( \rho^b(n) \big)^2 }{\theta  P_0^{a}(n) + \big(1 - \theta \big)  P_0^{b}(n) }}\\
\label{l2_eq4}
1-  \rho_2(n)  &=&  \sqrt{\frac{\theta  P_0^{a}(n) \big(1-  \rho^a(n)\big)^2 + (1 - \theta)  P_0^{b}(n) \big(1- \rho^b (n) \big)^2 }{\theta  P_0^{a}(n)  + (1 - \theta)  P_0^{b}(n) }}
\end{eqnarray}
By using the concavity of log function, the following can be proved,
\begin{eqnarray}
\theta B^a_1  + (1 - \theta)B^b_1 & \leq &  \theta  \sum_{n = 1}^{N} C \big(P_1^{a}(n)\big) + (1 - \theta) \sum_{n = 1}^{N} C \big(P_1^{b}(n) \big)  \nonumber \\
\label{l2_eq5}
&\leq &   \sum_{n = 1}^{N} C \big( \theta P_1^{a}(n) + (1 - \theta) P_1^{b}(n) \big) =  \sum_{n = 1}^{N} C \big(   P_1(n) \big) \, \, \, \,    \, \, \,
\end{eqnarray}
Similarly, the following can be proved,
\begin{eqnarray}
\label{l2_eq6}
\theta  B^a_2   + (1 - \theta)B^b_2   & \leq &  \sum_{n = 1}^{N} C \big(   P_2(n) \big) \\
\label{l2_eq7}
\sum_{i = 1}^{2} \theta B^a_i  + (1 - \theta)B^b_i   & \leq &  \sum_{n = 1}^{N} C \big( P_1(n ) +  P_2(n) \big) \\
\label{l2_eq8}
\sum_{i = 0}^{2} \theta B^a_i  + (1 - \theta)B^b_i   & \leq &  \sum_{n = 1}^{N} C \big(  P_0(n ) +P_1(n ) +  P_2(n) \big) 
\end{eqnarray}
Then, by using (\ref{l2_eq1}-\ref{l2_eq4}), the following can be proved,
\begin{eqnarray}
 \sum_{n = 1}^{N} \bigg(P_1(n) +  \rho^{2}_1(n) P_0(n) \bigg) L(n)   =  \, \, \, \theta \sum_{n = 1}^{N}  \bigg( P_{1}^{a}(n) + \big( \rho^{a}(n)\big)^2 P_0^{a}(n) \bigg) L(n) &&   \nonumber \\   
   +  \,  (1 - \theta) \sum_{n = 1}^{N} \bigg( P_{1}^{b}(n) + \big( \rho^{b}(n) \big)^2 P_0^{b}(n) \bigg) L(n)    \, &&   \nonumber \\
   \leq   \, \, \,  \theta \widetilde{E_{1, a}}(N)  + (1 - \theta)  \widetilde{E_{1, b}}(N)  \, \, \, \,    \, \, \, \,  \, \, \, \,  \, \, \, \, \, \, \, \, \, \, \, \, \, && 
   \label{l2_eq9}
\end{eqnarray}
Similarly, it can be proved that 
\begin{equation}
\label{l2_eq10}
\sum_{n = 1}^{N} \bigg(P_2(n) +  \big(1- \rho^{2}_2(n)\big)^2 P_0(n)\bigg) L(n)  \leq   \theta \widetilde{E_{2, a}}(N)  + (1 - \theta)  \widetilde{E_{2, b}}(N)   
\end{equation}
It can be easily proved that $\big(1 - \rho_1(n)\big)^2 \leq \big(1 - \rho_2(n)\big)^2$ by using (\ref{l2_eq3}-\ref{l2_eq4}). Then, inserting into (\ref{l2_eq9}-\ref{l2_eq10}) and extending the result for the time interval $i$, the following equations can be obtained,
\begin{eqnarray}
\label{l2_eq11}
\sum_{n = 1}^{i} \bigg( P_1(n) +  \rho^{2}_1(n) P_0(n) \bigg) L(n)   \leq     \theta \widetilde{E_{1, a}}(i)  + (1 - \theta)  \widetilde{E_{1, b}}(i)  &&   \\
\label{l2_eq12}
\sum_{n = 1}^{i} \bigg(P_2(n) +  \big(1- \rho^{2}_1(n) \big)^2 P_0(n) \bigg) L(n)  \leq    \theta \widetilde{E_{2, a}}(i)  + (1 - \theta)  \widetilde{E_{2, b}}(i)   &&
\end{eqnarray}
As a result, $\theta \boldsymbol{B^a}  + (1 - \theta)\boldsymbol{B^b} \in   B(\boldsymbol{P_1}, \boldsymbol{P_2}, \boldsymbol{P_0}, \boldsymbol{\rho_1})$. Therefore,  $\theta \boldsymbol{B^a}  + (1 - \theta)\boldsymbol{B^b} \in$  $B_d( \theta \boldsymbol{\widetilde{E_{1, a}}}  + (1 - \theta) \boldsymbol{ \widetilde{E_{1, b}}},$  $\theta \boldsymbol{\widetilde{E_{2, a}}}  + (1 - \theta) \boldsymbol{ \widetilde{E_{2, b}}}, N)$ and $\Big( \theta \boldsymbol{B^a}  + (1 - \theta)\boldsymbol{B^b}, $ $\theta \boldsymbol{\widetilde{E_{1, a}}} + (1 - \theta)\boldsymbol{\widetilde{E_{1, b}}},$ $  \theta \boldsymbol{\widetilde{E_{2, a}}} + (1 - \theta)\boldsymbol{\widetilde{E_{2, b}}}   \Big)$ is in  $\mathbb{P}$.
\end{proof}

\begin{proof}[Proof of Theorem \ref{lemma3}]

The optimum solution of $\left( P_0(n), P_1(n), P_2(n) \right)$ regarding the single time step $n$ with the available power levels $\overline{P_1}(n)$ and $\overline{P_2}(n)$ for that iteration is found by using the iterative descent algorithm defined in \cite{b2} with 8 different Lagrange multipliers regions and the corresponding solutions as shown in Table \ref{tab_app1} in  Appendix \ref{app1}. There is a difference between the optimum solution for the single time step between \cite{b2} and the current study. The Lagrange multipliers obtained in the solution for the single time step are the solutions for the power inequalities corresponding to $\overline{P_1}(n)$ and $\overline{P_2}(n)$ rather than (\ref{c1_eq}-\ref{c5_eq2}) in the proposed solution. Therefore, the multipliers $\lambda_{1,n}^{p}$ and $\lambda_{2,n}^{p}$ include the effects of the all time intervals.

Compared with \cite{b4, b11}, whenever the water levels between time intervals $n$ and $n+1$ are equalized, there can be  potentially at most $8 \times 8 = 64$ different combinations of optimality regions defined in Table \ref{tab_app1}. Furthermore, in the optimum scheduling solution, there can be 3 different cases corresponding to the transfer of energy between two time intervals, i.e., power transfer to the next time interval by the 1st node, 2nd node and both of the nodes, complicating the analysis more. A large set of optimality relations exist between the water levels $W_{1,2,3}(n)$ and $W_{1,2,3}(n+1)$ based on (\ref{eq50}-\ref{eq52}). For example, some of the multipliers will be zero whenever the corresponding power levels $P_{0,1,2}(n)$ are greater than zero. On the other hand, $\lambda_{1,n}^{p}$ and $\lambda_{2,n}^{p}$ are already decreasing functions based on their definition, i.e., $\lambda_{1,n}^{p} \geq \lambda_{1,n+1}^{p}$, and the same for the $\lambda_{2,n}$.  $\lambda_{1,n}^{p}$ is equal to $\lambda_{1,n+1}^{p}$ if the 1st node transfers the stored energy from the time interval $n$ to $n+1$, and similarly $\lambda_{2,n}^{p}$ is equal to $\lambda_{2,n+1}^{p}$ for the 2nd node and both the equalities hold if both the nodes transfer energy to the next time interval. Moreover, it is observed that $g ( \lambda_{1,n}^{p},  \lambda_{2,n}^{p}) \geq g (  \lambda_{1,n+1}^{p},   \lambda_{2,n+1}^{p})$.  As a result, by using these set of observations and the equalities in (\ref{eq50}-\ref{eq52}), it can be observed that $W_{i}(n) \geq W_{i}(n+1)$ for $i \in [1,N]$ is always satisfied possibly leading to $W_{i}(n) = W_{i}(n+1)$ for some specific combinations of optimality regions between neighbouring time intervals $n$ and $n+1$. Due to space limitations, the result is not given for all different optimality region combinations, however, they can be shown easily.

For the comparison regarding $W_{4,5}(n)$ and $W_{4,5}(n+1)$, the fact that $\lambda_{i,n}^{p} = \lambda_{i,n+1}^{p}$ if the $i$th node transfers stored energy from the time interval $n$ to $n+1$ is utilized for $i \in [1,2]$. In fact, $\lambda_{i,n}^{p}$ is represented in terms of $\big(P_0(n)$, $P_1(n)$, $P_2(n)$, $\overline{P_1}(n)$, $\overline{P_2}(n)\big)$ in the corresponding time interval $n$. Therefore, instead of equalizing Lagrange multipliers $\lambda_{i,n}^{p}$ and $\lambda_{i,n+1}^{p}$, new water levels $W_{4}(n)$ and $W_{5}(n)$ represented in terms of the power levels are defined and equalized. 
\end{proof}

\begin{proof}[Proof of Theorem \ref{lemma4}]
When the 2nd node has fixed energy in one of iterations, the only water levels to be compared are the pairs ($W_1(n)$, $W_1(n+1)$) and ($W_4(n)$, $W_4(n+1)$) for $n \in [1,N-1]$. In the same manner, when the 1st node has fixed energy in one of iterations, the only water levels to be compared are the pairs ($W_2(n)$, $W_2(n+1)$) and ($W_5(n)$,$W_5(n+1)$) for $n \in [1,N-1]$. These can be proved by removing the inequalities from (\ref{eq50}-\ref{eq52})  including the corresponding multipliers regarding the fixed node. It can be easily proved by comparing the multipliers for the two neighbouring time intervals  such that the inequalities and equalities in Table \ref{tab0_app0} hold for the optimum solution. Then, it can be easily proved by using the solutions in Table \ref{tab_app1} such that equalizing ($W_4(n)$, $W_4(n+1)$) for $n \in [1,N-1]$ satisfies the inequalities or the equalities regarding ($W_1(n)$, $W_1(n+1)$) and ($W_2(n)$, $W_2(n+1)$) in Table \ref{tab0_app0}. Therefore, there is no need to compare the 1st and 2nd water levels and it is enough to equalize ($W_4(n)$, $W_4(n+1)$) and ($W_5(n)$,  $W_5(n+1)$) between neighbouring time intervals.

\renewcommand{\arraystretch}{1.7}
\begin{table*}[t!]
\caption{The optimality conditions for ($W_{1,2,4,5}(n)$, $W_{1,2,4,5}(n+1)$) }
\begin{center}
\footnotesize
\begin{tabular}{| M{1.3cm}|M{1.3cm}|M{2.7cm}|M{2.7cm}|}
\hline
   \multicolumn{4}{|c|}{$ \mu_1   \geq \mu_2  \geq \mu_0  $, 1st node power scheduling}    \\\hline 
Region Indices $(n)$    & Region Indices $(n+1)$ & \multicolumn{2}{c|}{Cond.}      \\\hline  
 $2$, $4$, $6$    & $2$, $4$, $6$    & \multirow{4}{*}[-0.3em]{$ W_{4}(n) = W_{4}(n+1) $} &     \\    \cline{1-2} \cline{4-4}  
 $2$, $4$, $6$    &  $3$, $5$, $7$, $8$    &  &  $ W_{1}(n) \geq W_{1}(n+1)$   \\   \cline{1-2} \cline{4-4}   
$3$, $5$, $7$, $8$    & $2$, $4$, $6$     &   &   $ W_{1}(n) \leq W_{1}(n+1)$  \\      \cline{1-2} \cline{4-4}   
$3$, $5$, $7$, $8$    &  $3$, $5$, $7$, $8$    &  &  $ W_{1}(n) = W_{1}(n+1)$   \\\hline
 \hline
    \multicolumn{4}{|c|}{$ \mu_1   \geq \mu_2  \geq \mu_0  $, 2nd node power scheduling}    \\\hline 
Region Indices $(n)$    & Region Indices $(n+1)$ & \multicolumn{2}{c|}{Cond.}      \\\hline  
 $2$, $3$, $5$    & $2$, $3$, $5$    & \multirow{4}{*}[-0.3em]{$ W_{5}(n) = W_{5}(n+1) $} &     \\    \cline{1-2} \cline{4-4}  
 $2$, $3$, $5$    &  $4$, $6$, $7$, $8$    &  &  $ W_{2}(n) \geq W_{2}(n+1)$   \\   \cline{1-2} \cline{4-4}   
$4$, $6$, $7$, $8$    & $2$, $3$, $5$     &   &   $ W_{2}(n) \leq W_{2}(n+1)$  \\      \cline{1-2} \cline{4-4}   
$4$, $6$, $7$, $8$    &  $4$, $6$, $7$, $8$    &  &  $ W_{2}(n) = W_{2}(n+1)$   \\\hline
\end{tabular}
\end{center}
\label{tab0_app0}
\end{table*}
\renewcommand{\arraystretch}{1}
\end{proof}

\section{}
\label{app1} 
The following definitions are used next, $\lambda_1 \equiv \lambda_{1,n}^{p}$, $\lambda_2 \equiv \lambda_{2,n}^{p}$, $P_i \equiv P_i^{'}(n)$, $g \equiv g(\lambda_1, \lambda_2)$, $\alpha \equiv  (\mu_2 - \mu_0) \, / \, (\lambda_2 -  g)$, $\beta \equiv (\mu_1 - \mu_2) \, / \, (\lambda_1 - \lambda_2)$, $\gamma \equiv (\mu_0 - \mu_1)(g - \lambda_1)$. The regions are defined for varying $(\lambda_1 , \lambda_2)$  with the corresponding solutions. $R_i$ refers the region, and $S_i$ refers the solution in that region where $R_1$ refers to tzero power region. KKT or Lagrange multiplier regions and the solutions for the optimum power levels are given in Table \ref{tab_app1} resembling the structure in \cite{b2}.
\renewcommand{\arraystretch}{1.7}
\begin{table*}[t!]
\caption{The optimality regions and the solutions for the single  time interval}
\begin{center}
\footnotesize
\begin{tabular}{| M{0.6cm}|m{5.5cm}|m{6.5cm}|}
\hline
   \multicolumn{3}{|c|}{$ \mu_1   \geq \mu_2  \geq \mu_0  $}    \\\hline 
 Index    & Region & Solution    \\\hline 
 $1$ & $  \mu_i < \lambda_i: i \in [1,2], g > \mu_0 $               & $ P_{0,1,2}= 0 $    \\\hline 
 $2$ & $   \mu_i - \mu_0 + g < \lambda_i: i \in [1,2], \frac{\mu_0}{g}  > 1  $            & $  P_0 = \frac{\mu_0}{g} - 1, \, P_{1,2} = 0$    \\\hline 
  $3 $          &   $    \mu_1  > \lambda_1 , \frac{g}{\lambda_1}  >  \frac{\mu_0}{\mu_1} ,     \frac{\lambda_2}{\lambda_1 }   >   \frac{\mu_2}{\mu_1}   $     & $ P_{0,2} = 0, \, P_1 = \frac{\mu_1}{\lambda_1} - 1 $    \\\hline 
 $4 $          &   $   \mu_1 - \mu_2 < \lambda_1 - \lambda_2,  \frac{g}{\lambda_2 }  >   \frac{\mu_0}{\mu_2} ,     \lambda_2  <  \mu_2   $     & $  P_{0,1} = 0, \,  P_2 = \frac{\mu_2}{\lambda_2} - 1  $    \\\hline 
  $ 5$          &   $    \frac{\mu_0}{g}  > 1 , \gamma   >  \alpha ,  P_{0,1} > 0 $     & $ P_0 = \frac{\mu_0 }{g} - P_1  - 1, P_1 = \gamma - 1   $    \\\hline 
   $6 $          &   $    \frac{\mu_0}{g}  > 1 , \frac{\lambda_1 - \lambda_2}{\mu_1 - \mu_2}  > 1, \alpha > 1, P_0 > 0   $     & $ P_0 = \frac{ \mu_0}{ g} - P_2 - 1,  P_{1} = 0, P_2 = \alpha - 1  $    \\\hline 
  $7 $          &   $   \frac{g}{\lambda_2}  >   \frac{\mu_0}{\mu_2}, \frac{\mu_2}{\lambda_2} > 1, \frac{\lambda_2}{\lambda_1} <  \frac{\mu_2}{\mu_1},  P_1 > 0  $     & $P_{0} = 0, P_1 = \frac{\mu_1 - \mu_2}{\lambda_1 -   \lambda_2 } - 1 ;  P_2 = \frac{\mu_2}{\lambda_2} - P_1  - 1 $    \\\hline 
   $ 8$          &   $    \frac{\mu_0}{g} >  \beta  > 1,   \alpha > 1, P_{0, 2} > 0  $     & $P_0 = \frac{\mu_0 }{g} - \frac{\mu_2 - \mu_0}{\lambda_2 - g}, P_1 = \frac{\mu_1 - \mu_2}{\lambda_1 -   \lambda_2 } - 1, \newline  P_2 = \frac{\mu_2 - \mu_0}{\lambda_2 - g}  - P_1  - 1 $      \\  \hline \hline
  \multicolumn{3}{|c|}{$ \mu_1   \geq \mu_0  \geq \mu_2  $}    \\\hline 
$1$ & $   \mu_1 < \lambda_1,  g > \mu_0 $               & $ P_{0,1,2}= 0 $    \\\hline 
 $2$ & $   \mu_1 > \lambda_1 , \frac{g}{\lambda_1} >  \frac{\mu_0}{\mu_1} $            & $ P_1 = \frac{\mu_1}{\lambda_1} - 1, \, P_{0,2} = 0$    \\\hline 
  $3 $          &   $  \frac{1 }{\gamma} > 1 ,     \frac{\mu_0}{g }> 1  $     & $  P_0 = \frac{\mu_0}{g} - 1 , \, P_{1,2} = 0 $    \\\hline 
 $4 $          &   $  \frac{\mu_0}{g} > 1, P_{0, 1} > 0   $     & $  P_1 = \frac{\mu_1 - \mu_0}{\lambda_1 - g}  - 1 , \, P_0 = \frac{\mu_0}{g} - P_1 - 1 , \newline P_{2} = 0    $    \\\hline 
 \hline 
   \multicolumn{3}{|c|}{$ \mu_0   \geq \mbox{max}(\mu_1, \mu_2 )$}    \\\hline 
   $1$          &   $  \frac{\mu_0}{g} <= 1 $     & $  P_{0,1,2}= 0   $    \\\hline 
   $2$          &   $  \frac{\mu_0}{g} > 1 $     & $  P_0 = \frac{\mu_0}{ g}  - 1,  P_{1,2}= 0    $    \\\hline 
\end{tabular}
\end{center}
\label{tab_app1}
\end{table*}
\renewcommand{\arraystretch}{1}

\end{appendices}

%
%




%
%

\end{document}